\documentclass[DIV=16]{scrartcl}

\usepackage[utf8]{inputenc}
\usepackage[font=normalsize]{caption}
\usepackage[font=normalsize]{subcaption}
\usepackage{authblk}
\usepackage{amsthm}
\usepackage{tabularx}
\usepackage{booktabs}

\newtheorem{theorem}{Theorem}

\usepackage{mathtools, bm, amsmath, amsfonts, amsthm, amssymb, amscd}

\DeclareMathAlphabet{\mathcalON}{OT1}{pzc}{m}{n}

\newcommand \s {{\mathrm s}}
\renewcommand \L {{\mathrm L}}
\newcommand \dd {{\mathrm d}}

\DeclareBoldMathCommand \cX{{\cal X}}
\DeclareBoldMathCommand \cA{{\cal A}}
\DeclareBoldMathCommand \cB{{\cal B}}
\DeclareBoldMathCommand \cU{{\cal U}}
\DeclareBoldMathCommand \cJ{{\cal J}}
\newcommand \cP {\mathbf{p}}
\DeclareBoldMathCommand \cQ{{\cal Q}}
\DeclareBoldMathCommand \cC{{\cal C}}
\newcommand \cW {\mathbf{w}}
\DeclareBoldMathCommand \cI{{\cal I}}

\newcommand \bfw {\mathbf{w}}

\newcommand \bfp {\mathbf{p}}

\newcommand \bfc {\mathbf{c}}

\newcommand \bfx {\mathbf{x}}

\newcommand \bfA {\mathbf{A}}

\newcommand \bfB {\mathbf{B}}

\DeclareBoldMathCommand \cItil{{\cal \widetilde{I}}}
\DeclareBoldMathCommand \cAtil{{\cal \widetilde{A}}}

\DeclareBoldMathCommand \cQtil{{\cal \widetilde{Q}}}

\newcommand \bfxhat {\mathbf{\widehat{x}}}

\newcommand \bfchat {\mathbf{\widehat{c}}}

\newcommand \xhat {{\widehat{x}}}

\newcommand \bSpl {{P}}
\newcommand \vihat {\widehat{v}_{\mathrm i}}

\newcommand \degree {{\bar{p}}}
\newcommand \vibar {\bar{v}_{\mathrm i}}
\newcommand \vCbarDesired {\bar{v}_{\mathrm{C},\text{desired}}}

\newcommand \Np {N_{\mathrm p}}

\newcommand \Ns {N_{\mathrm s}}

\newcommand \Ts {T_{\mathrm s}}
\newcommand \fs {f_\mathrm{s}}
\newcommand \iL {i_\mathrm{L}}
\newcommand \vC {v_{\mathrm C}}

\newcommand \vi {v_{\mathrm i}}

\newcommand \du {{d}} 
\newcommand \Tac {{T_{\mathrm{ac}}}}
\newcommand \fac {{f_{\mathrm{ac}}}}

\newcommand \Ltwo {{\mathrm{L}^2}}

\title{Efficient simulation of DC-AC power converters using Multirate Partial Differential Equations}

\author[1,2,3]{Andreas Pels}
\author[3]{Ruth V. Sabariego}
\author[1,2]{Sebastian Schöps}
{
  \affil[1]{\small Graduate School of Computational Engineering, Technische Universität Darmstadt, Dolivostraße 15, 64293 Darmstadt, Germany}
  \affil[2]{Institut für Teilchenbeschleunigung und Elektromagnetische Felder, Technische Universit\"at Darmstadt, Schlossgartenstraße 8, 64289 Darmstadt, Germany}
  \affil[3]{Department of Electrical Engineering, EnergyVille, KU Leuven, Kasteelpark Arenberg 10, 3001 Leuven, Belgium}
}

\date{}
\begin{document}
\maketitle

\begin{abstract}
  Switch-mode power converters are used in various applications to convert between different voltage (or current) levels. They use transistors to switch on and off the input voltage to generate a pulsed voltage whose arithmetic average is the desired output voltage of the converter. After smoothening by filters, the converter output is used to supply devices. The simulation of these switch-mode power converters by conventional time discretization is computationally expensive since a high number of time steps is necessary to properly resolve the unknown state variables and detect switch events of the excitation. This paper proposes a multirate method based on the concept of Multirate Partial Differential Equations (MPDEs), which splits the solution into fast varying and slowly varying parts. The method is developed to work with pulse width modulated (PWM) excitation with a constant switching cycle and varying duty cycle. The important case of varying duty cycles in the MPDE framework is adressed for the first time. Switching event detection is no longer necessary and a much smaller number of time steps for a decent resolution are required, thus leading to a highly efficent method.
\end{abstract}

\section{Introduction}\label{sec1}
Switch-mode power converters play nowadays a vital role in various applications \cite{Mohan_2003aa, Vasca_2012aa}. They are used in domestic devices, e.g., in mobile phone chargers and computer power supply, as well as in industrial applications, e.g., in high-voltage DC (HVDC) power transfer and speed control of electrical motors, to convert voltage (or current) between different levels. Switch-mode power converters use transistors to periodically switch on and off the input voltage to generate a pulsed voltage, whose arithmetic average is the desired output of the converter. The pulsed voltage is smoothened by filter circuits, which also act as energy buffer before it is used to supply the appliance.

An exemplary circuit of such a switch-mode power converter is depicted in Fig.~\ref{fig:buckConverterAndSol} along with its solution in DC-DC conversion mode. As can be seen the converter output voltage, even after filtering, still comprises of high-frequency components generated by the pulsed excitation, which are usually referred to as ripples \cite{Vasca_2012aa}. Since the circuit is uncharged at the beginning, the converter needs some time to reach the steady state. As a result, the solution consists of a slowly varying envelope which is modulated with the fast periodically varying ripples. The method used to generate the control signals for the transistors is called pulse width modulation (PWM). There are different kinds of PWM \cite{Vasca_2012aa}, e.g., natural sampling with different carrier signals like sawtooth and triangle, and regular sampling. Since in this paper, we focus on natural sampling with a sawtooth carrier, we refer the interested reader to Vasca et al. \cite{Vasca_2012aa} for more details on other modulations. Using natural sampling with a sawtooth carrier (trailing edge) leads to a pulsed signal as depicted in Fig.~\ref{fig:pulsedExcitation}, where the switching (pulse) period $\Ts=1/\fs$ and the duty cycle $\du$ are the quantities defining the pulses. The switching frequency is constant while the duty cycle varies with time. 

\begin{figure}
  \centering
  \begin{subfigure}[c]{1\textwidth}
    \centering
    \includegraphics{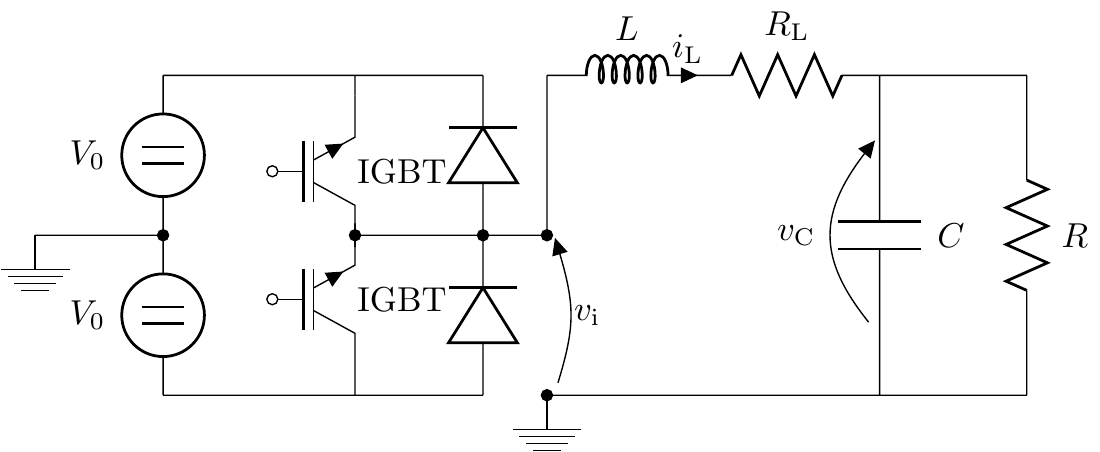}
  \end{subfigure}
  
  \vspace{1em}
  \begin{subfigure}[c]{1\textwidth}
    \centering
    \includegraphics{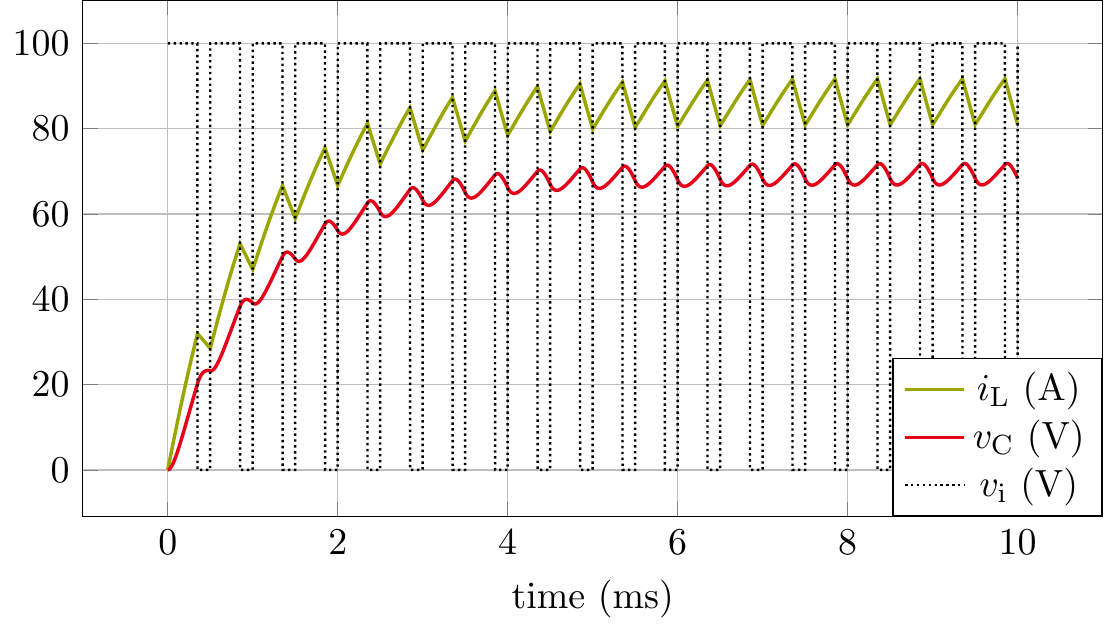}
  \end{subfigure}
  \caption{Exemplary circuit of a switch-mode power converter and its solution at a switching frequency of $\fs=2\,$kHz and constant duty cycle $\du=0.7$.}
  \label{fig:buckConverterAndSol}
\end{figure}

\begin{figure}
  \centering
  \includegraphics{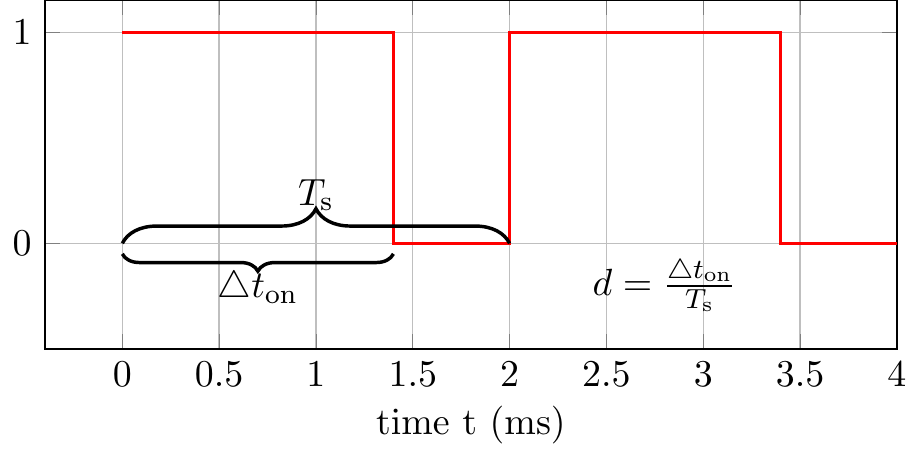}
  \caption{Pulsed excitation and the quantities defining it: switching period $\Ts$, on interval $\triangle t_\mathrm{on}$, duty cycle $\du$.}
  \label{fig:pulsedExcitation}
\end{figure}

Large-scale simulations of power converters require simplifications, e.g., ideal switching behaviour of the transistors and ideal diode \cite{Mohan_2003aa}. Nonetheless, the simulation of power converters with conventional time discretization is still computationally expensive since a high number of time steps is necessary to properly resolve the ripples in the solution. Special techniques are commonly applied to detect switching events and up to now most software uses techniques which reduces the order of the applied time discretization to lowest order, see Tant et al. \cite{Tant_2018aa}. In this paper we develop a method which alleviates the need of high number of time steps by splitting the solution into fast and slowly varying parts. The resulting method is a (high-order) multirate approach, which is based on a concept called Multirate Partial Differential Equations (MPDEs) \cite{Brachtendorf_1996aa, Roychowdhury_2001aa}. For DC-DC power converters the method has already been proposed in earlier papers by Pels et al. \cite{Pels_2018aa, Pels_2017ad}. We extend the concept to DC-AC power converters (also called inverters), in the following. In the process we restrict ourselves to power converter circuits which can be described by linear electrical elements and thus by a system of linear ordinary differential equations (ODEs). An extension to nonlinear problems can be achieved similarly as described in Pels et al. \cite{Pels_2018aa}. The method is verified using the example of the converter in Fig.~\ref{fig:buckConverterAndSol}, which is referred to as inverter in the following since it is operated in DC-AC mode. The main contribution is the efficient treatment of (sinusoidally) varying duty cycle. 

The paper is structured as follows. Section \ref{sec2} introduces the concept of MPDEs and their relation to the original ODEs describing the application. The focus of Section \ref{sec3} lies on the modeling of the DC-AC converters using MPDEs. It describes how an efficient simulation can be achieved. The methods used for solving the MPDEs, i.e., a Galerkin approach for fast periodic parts of the solution and a conventional time discretization for the slowly varying parts are presented. Section \ref{sec4} introduces B-spline basis functions used for the solution expansion and discusses their properties. It is shown that the basis functions are well suited for use with the proposed method since they allow to exploit smoothness almost everywhere but still lead to a cheap assembly of the arising matrices. Finally in Section \ref{sec5}, the method is applied to the example of the inverter and the computational efficiency is discussed. Section \ref{sec6} concludes the paper. 

\section{Introduction to Multirate Partial Differential Equations}\label{sec2}
Let the mathematical model of the converter be described by linear ordinary differential equations in the form
\begin{equation}
  \bfA \frac{\dd }{\dd t} \bfx(t) + \bfB \bfx(t) = \bfc(t), \quad \bfx(0)=\bfx_0,
  \label{equ:origODEs}
\end{equation}
where $\bfA, \bfB \in \mathbb{R}^{\Ns \times \Ns}$ are matrices, which depend on the topology of the circuit and the electrical elements, $\bfx(t) \in \mathbb{R}^{\Ns}$ is the unknown solution consisting of voltages and currents, $\bfx_0$ is the vector of initial conditions and $\bfc(t)\in \mathbb{R}^{\Ns}$ is the excitation of the circuit.   

To obtain the multirate formulation, two artifical time scales $t_1$ and $t_2$ are introduced and \eqref{equ:origODEs} is rewritten in terms of the corresponding multivariate solution $\bfxhat(t_1,t_2)$ and excitation $\bfchat(t_1,t_2)$ which yields the MPDEs \cite{Brachtendorf_1996aa, Roychowdhury_2001aa}
\begin{equation}
  \bfA \left(\frac{\partial \bfxhat(t_1,t_2)}{\partial t_1} + \frac{\partial \bfxhat(t_1,t_2)}{\partial t_2}\right) + \bfB \bfxhat(t_1, t_2) = \bfchat(t_1,t_2).
  \label{equ:origMPDEs}
\end{equation}
The relation between the ODEs \eqref{equ:origODEs} and MPDEs \eqref{equ:origMPDEs} is given by \cite{Brachtendorf_1996aa, Roychowdhury_2001aa}
\begin{equation}
  \begin{aligned}
    \bfxhat(t, t)&=\bfx(t) \\
    \bfchat(t, t)&=\bfc(t).
  \end{aligned}
  \label{equ:relationMPDEsODEs}
\end{equation}
Therefore, if any right-hand side $\bfchat(t_1,t_2)$ can be found which fulfills the relation $\bfchat(t,t)=\bfc(t)$, then the solution of the original system of ODEs can be extracted from the solution of the MPDEs by using $\bfx(t)=\bfxhat(t,t)$, i.e., evaluating the multivariate solution along a diagonal line through the computational domain. It is important to note that the simulation efficiency depends on the choice of $\bfchat(t_1,t_2)$ and of course the methods which are used for solving the MPDEs. In our case, we require fast changes to occur along the ``fast time scale'' $t_2$ and slow changes along the ``slow time scale'' $t_1$. The multivariate right-hand side $\bfchat(t_1,t_2)$ has then to be chosen appropriately, which depends on the application and is detailed in Section \ref{sec5}. 

\section{Multirate Modeling of DC-AC converters}\label{sec3}
The fast varying components, i.e., the ripples in the solution of the power converters, are assumed to be representable by basis functions $p_k(\tau(t_2), \du(t_1))$, which depend on the fast time scale $t_2$ and on the duty cycle $\du(t_1)$ assumed to be slowly varying. The first assumption is necessary for the convergence of the method but obvious for most classical bases. The second assumption on the dynamics of the duty cycle is not crucial but will be important to obtain an efficient method. The basis functions are periodic (with period $\Ts$), which is realized using the function $\tau(t_2) = \frac{t_2}{\Ts}\text{ modulo } 1$, also called the relative time. The multivariate solution is expanded into the basis functions and slowly varying coefficients $w_{j,k}(t_1)$ yielding \cite{Pels_2017ad,Pels_2018aa,Gyselinck_2013ab}
\begin{equation}
  \xhat_j(t_1, t_2)\doteq\sum\limits_{k=0}^{\Np} w_{j,k}(t_1) p_k(\tau(t_2), \du(t_1)) = \bfp^\top\!(\tau(t_2), \du(t_1)) \bfw_j(t_1), 
  \label{equ:solExp}
\end{equation}
where 
\begin{equation}
  \cP(\tau) =  
  \left[\begin{array}{c} p_0 \\ p_1(\tau) \\ p_2(\tau) \\ \vdots \\  p_{N_{\mathrm p}}(\tau)   \end{array}\right] \ , \ \
  \quad
  \cW_j(t_1) \ = \ 
  \left[\begin{array}{c} w_{j,0}(t_1) \\ w_{j,1}(t_1) \\ w_{j,2}(t_1) \\ \vdots \\  w_{j,\Np}(t_1)   \end{array}\right]
  \, .
  \vspace{1em}
\end{equation}
Inserting the solution expansion into the partial derivatives from \eqref{equ:origMPDEs} leads to
\begin{align}
  \frac{\partial \xhat_j(t_1, t_2)}{\partial t_1}+\frac{\partial \xhat_j(t_1, t_2)}{\partial t_2} &= \frac{\partial \bfp^\top\!(\tau(t_2), \du)}{\partial \, \du} \frac{\dd \, \du(t_1)}{\dd t_1} \bfw_j(t_1) \nonumber \\
  &+ \bfp^\top\!(\tau(t_2), \du(t_1)) \frac{\dd \bfw_j(t_1)}{\dd t_1} + \frac{\partial \bfp^\top\!(\tau(t_2), \du(t_1))}{\partial \tau} \frac{\dd \tau(t_2)}{\dd t_2} \bfw_j(t_1),
  \label{equ:partDerivExtended}
\end{align}

To solve the MPDEs, two methods are applied, one method for each time scale. To resolve the fast changes represented by the basis functions $p_k$, a Galerkin approach is applied. The slowly varying coefficients are solved afterwards with a conventional time discretization. Applying the Ritz-Galerkin approach to the MPDEs \eqref{equ:origMPDEs} with respect to $t_2$ and on the interval of periodicity $[0,\Ts)$ yields 
\begin{align}
  &\int\limits_0^{\Ts} \left(\bfA \left(\frac{\partial \bfxhat(t_1,t_2)}{\partial t_1} + \frac{\partial \bfxhat(t_1,t_2)}{\partial t_2}\right) + \bfB \bfxhat(t_1, t_2)\right) p_l(\tau(t_2), \du) \dd t_2 \nonumber \\
  &= \int\limits_0^{\Ts} \bfchat(t_1,t_2) p_l(\tau(t_2), \du) \dd t_2 \quad \forall l=0,\dots,\Np.
\end{align}
Using integration by parts gives (function arguments are omitted in the following for the sake of readability)
\begin{align}
  \int\limits_0^{\Ts} \bfA \frac{\partial \bfxhat(t_1,t_2)}{\partial t_1} &+ \bfB \bfxhat(t_1, t_2) \, p_l \, \dd t_2 - \int\limits_0^{\Ts} \bfA \bfxhat(t_1,t_2)\frac{\partial p_l}{\partial \tau} \frac{\dd \tau}{\dd t_2} \dd t_2 \nonumber \\
  &+ \Big(\bfxhat(t_1,t_2) p_l \Big)\Big|_{t_2=0}^{t_2=\Ts} 
  = \int\limits_0^{\Ts} \bfchat(t_1,t_2) \, p_l \, \dd t_2 \qquad \forall l=0,\dots,\Np.
\end{align}
Since the solution is periodic with respect to $t_2$ in the interval $[0,\Ts)$, the boundary term vanishes. 
Inserting \eqref{equ:partDerivExtended} and substituting $t_2$ with $\tau$ leads to
\begin{equation}
  \begin{aligned}
    &\int\limits_0^{1} \left(\bfA 
    \left[\begin{array}{c} 
      \frac{\partial \bfp^\top}{\partial \, \du} \frac{\dd \, \du}{\dd t_1} \bfw_1 \\
      \vdots \\
      \frac{\partial \bfp^\top}{\partial \, \du} \frac{\dd \, \du}{\dd t_1} \bfw_{\Ns}
    \end{array}\right]
    + \bfA 
    \left[\begin{array}{c} 
      \bfp^\top \frac{\dd \bfw_1}{\dd t_1} \\
      \vdots \\
      \bfp^\top \frac{\dd \bfw_{\Ns}}{\dd t_1}
    \end{array}\right]
    + 
    \bfB 
    \left[\begin{array}{c} 
      \bfp^\top \bfw_1 \\
      \vdots \\
      \bfp^\top \bfw_{\Ns}
    \end{array}\right]
    \right) p_l \, \Ts \dd \tau \\
    &-
    \int\limits_0^{1} \left(\bfA 
    \left[\begin{array}{c} 
      \bfp^\top \bfw_1 \\
      \vdots \\
      \bfp^\top \bfw_{\Ns}
    \end{array}\right] \right) \frac{\partial p_l}{\partial \tau} \underbrace{\frac{1}{\Ts} \,  \Ts}_{=1} \dd \tau 
    = 
    \int\limits_0^{\Ts} \bfchat(t_1,t_2) p_l \, \dd t_2 \qquad \forall l=0,\dots,\Np.
  \end{aligned}
\end{equation}
Assembling all equations finally yields a linear system of ODEs with time-varying coefficients
\begin{equation}
  \cA(t_1) \, \frac{\dd  \bfw}{\dd t_1} + \cB(t_1) \, \bfw(t_1) \ = \ \cC(t_1) \, ,
  \label{equ:reducedMPDEs}
\end{equation}
where 
\begin{align}
  \cA(t_1) &=\bfA\otimes\cI(\du(t_1)), \quad \cB(t_1)=\bfB\otimes\cI(\du(t_1))+\bfA\otimes\cQ(\du(t_1)) + \frac{\dd \, \du(t_1)}{\dd t_1} \bfA\otimes\cU(\du(t_1)),\\
  \cC(t_1)&=\int\limits_{0}^{T_\s}
  \bfchat(t_1,t_2)  \otimes \cP(\tau(t_2), \du(t_1))  \, \dd t_2 \, .
  \label{equ:integralC}
\end{align}
and 
\begin{align}
  \cI(\du)  \ &= \ \Ts \int\limits_0^1  \cP(\tau, \du) \, \cP^\top \!(\tau, \du)  \,\dd \tau \, , \label{equ:matI} \\
  \cQ(\du) \ &= \ - \int\limits_0^1  \frac{\partial \cP(\tau, \du)}{\partial \tau } \, \cP^\top(\tau, \du)  \,  \dd \tau  \, , \label{equ:matQ} \\
  \cU(\du) \ &= \ \Ts \,  \int\limits_0^1  \cP(\tau, \du) \, \frac{\partial \bfp^\top\!(\tau, \du)}{\partial \, \du}  \,\dd \tau \, , \label{equ:matU}
\end{align}
The equation system \eqref{equ:reducedMPDEs} can be solved using conventional time discretization with much larger time steps than for the original problem, since fast varying changes are already taken into account by the Galerkin approach and only the envelope is resolved. A disadvantage of \eqref{equ:reducedMPDEs} is the larger equation systems, which are $\Np+1$ times larger than the original ones \eqref{equ:origODEs}. 

To solve the system \eqref{equ:reducedMPDEs} numerically, initial values have to be specified. To achieve an efficient simulation the following choice has proven advantageous: 
\begin{enumerate}
  \item Calculate the steady-state of the system \eqref{equ:reducedMPDEs}, i.e., $\left[\begin{array}{ccc} \bfw^{s}_1(0) & \dots & \bfw^{s}_{\Ns}(0)\end{array}\right]^\top
=\cB^{-1}(0) \cC(0)$ and use the solution expansion $\xhat^{s}_j(0,t_2)=\bfp^\top\!(\tau(t_2), \du(0)) \, \bfw^{s}_j(0)$, $j=1,\dots,\Ns$ to reconstruct the solution (ripple).
\item Since the reconstructed solution in steady-state does (usually) not fulfill the initial condition of the original ODEs \eqref{equ:origODEs}, i.e., $\bfxhat^{s}(0,0)\neq \bfx_0$, it has to be shifted by a constant $\bfc^\mathrm{s}$ as such that $\bfxhat(0,0)=\bfxhat^{s}(0,0)-\bfc^\mathrm{s}=\bfx_0$. The shift is accomplished by modifying the coefficients $\bfw^{s}(0)$. In case of B-splines, due to partition of unity, this is achieved by choosing the coefficients $\bfw_j(0)=\bfw^{s}_j(0)-c^{\mathrm{s}}_j$ as initial values for \eqref{equ:reducedMPDEs}, where $c^{\mathrm{s}}_j$ is the $j$-th component of $\bfc^\mathrm{s}$. This corresponds to the zero initial conditions as we use them for the original set of ODEs \eqref{equ:origODEs}. 
\end{enumerate}
Other choices may still lead to the correct solution but may require higher effort from the time discretization algortihm.

\section{Choice of basis functions}\label{sec4}
For the solution expansion \eqref{equ:solExp} B-spline basis functions are employed. They allow a high-order basis while still capturing the $C^0$ continuity as it appears in the current ripples by construction. In the literature splines have also been used in the context of MPDE methods but for high-frequency problems. Brachtendorf et al. \cite{Brachtendorf_2009aa} for instant use cubic and exponential splines, which lead to a better approximation than Fourier basis functions for steep transients, while still enabling a simple extraction of the frequency spectrum and Bittner et al. \cite{Bittner_2014aa} use an adaptive spline-wavelet to approximate solutions with steep transients. Both do not deal with $C^0$ solutions by construction. In contrast we take advantage of the a-priori knowledge of the excitation switching instant to contruct appropriate basis functions for the solution representation. 

\subsection{Introduction to B-splines}
In this subsection we briefly introduce the most important properties and definitions concerning B-splines for this work. The information is taken from Piegl et al. \cite{Piegl_1997aa}, to which the interested reader is referred for more details.

To define B-splines basis functions we first setup a knot vector $\Xi=\{\xi_0, \dots, \xi_m\}$, which is sorted in ascending order, i.e., $\xi_i \leq \xi_{i+1}$, $i=0,\dots, m-1$. 
The basis functions of degree $\degree$ are now build up recursively from the piecewise constant basis function $\degree=0$
\begin{equation}
  \bSpl_{i,0}(\xi)=
  \left\{
  \begin{array}{cc}
    1 & \text{for } \xi_i \leq \xi < \xi_{i+1} \\
    0 & \text{otherwise}
  \end{array}
  \right.
\end{equation}
by the Cox-DeBoor recurrence formula
\begin{equation}
  \bSpl_{i,\degree}(\xi)=\frac{\xi-\xi_i}{\xi_{i+\degree}-\xi_i} \bSpl_{i,\degree-1}(\xi) + \frac{\xi_{i+\degree+1}-\xi}{\xi_{i+\degree+1}-\xi_{i+1}} \bSpl_{i+1,\degree-1}(\xi).
  \label{equ:coxDeBoor}
\end{equation} 
For our purposes we assume the knot vector to be open (also called clamped), i.e., it is of the form
\begin{equation}
  \Xi=\{\underbrace{a,\dots,a}_{\degree+1}, \xi_{\degree+1}, \dots, \xi_{m-\degree-1}, \underbrace{b, \dots, b}_{\degree+1}\},
\end{equation}
where the first and the last knots appear $\degree+1$ times. 
The regularity $r_j$ of the B-spline basis functions across the knots, i.e., their continuity $C^{r_j}$ across the knot $\xi_j$, $j=0,\dots,m$, can be controlled using knot repetitions. The maximum regularity for degree $\degree$ B-splines is given by $r_{j,\mathrm{max}}=\degree-1$, which corresponds to knot repetitions of one for all knots except the knots at the boundary. To obtain less regularity, a knot repetition is introduced. Continuity $C^{r_j}$ is achieved by a knot repetition of $s_j=\degree-r_j$. For instant, to represent a $C^0$ continuity across knot $j$, the knot repetition is given by $s_j=\degree-0=\degree$. 

\subsection{Choice of knot vector}
For use in this work we define a knot vector as such, that the basis functions have a $C^0$ continuity at the point where the excitation changes its state. 
Defining the basis in terms of the relative time, i.e., $\tau\in[0,1]$, and the (fixed) duty cycle $\du\in(0,1)$ the simplest knot vector depending on the degree $\degree$ is given by 
\begin{equation}
  \Xi_{\degree}=\{\underbrace{0,\dots,0}_{\degree+1},\underbrace{\du,\dots,\du}_{\degree},\underbrace{1,\dots,1}_{\degree+1}\}.
\end{equation}
Additional knot refinement (corresponding to $h$-refinement in Finite Element Methods) leads to 
\begin{equation}
  \Xi_{\degree,K}=\{\underbrace{0,\dots,0}_{\degree+1},\alpha_1 \du, \dots, \alpha_K \du,\underbrace{\du,\dots,\du}_{\degree},\beta_1 (1-\du)+\du, \dots, \beta_K (1-\du)+\du,\underbrace{1,\dots,1}_{\degree+1}\}, 
  \label{equ:knotVectorRefined}
\end{equation}
where $K$ is the number of additional knots inserted before and after the $C^0$ continuity, and $\alpha_k \in [0,1]$, $\beta_k\in[0,1]$, $\forall k=1,\dots,K$ in ascending order. 

Using the refined knot vector \eqref{equ:knotVectorRefined} leads to the set of basis functions 
\begin{equation}
  \{\bSpl_{0,\degree}(\xi, \du), \dots,\bSpl_{2\degree+2K,\degree}(\xi, \du)\}
  \label{equ:BSplinesSpecific}
\end{equation}
as depicted exemplary for $K=1$ and $\degree=2$ in Fig.~\ref{fig:bSplines}, i.e., we have in total $2(\degree+K)+1$ basis functions. The basis functions for the solution expansion \eqref{equ:solExp} are finally given by
\begin{equation}
  p_k(\tau,\du)=\bSpl_{k,\degree}(\tau,\du)
\end{equation}
The periodicity of the set of basis functions is ensured using periodic boundary conditions in the implementation. 

Using the set of B-spline basis functions as introduced above leads to a cheap calculation of the matrices arising in the MPDE approach since their elements depend only linearly on the duty cycle. For a low-order basis like classical Finite Element hat functions (i.e., $\degree=1$) this is rather obvious whilst for higher order B-splines it is more involved due to their non-local support. Hence this property is analyzed in the following theorem.

\begin{figure}
  \centering
  \includegraphics{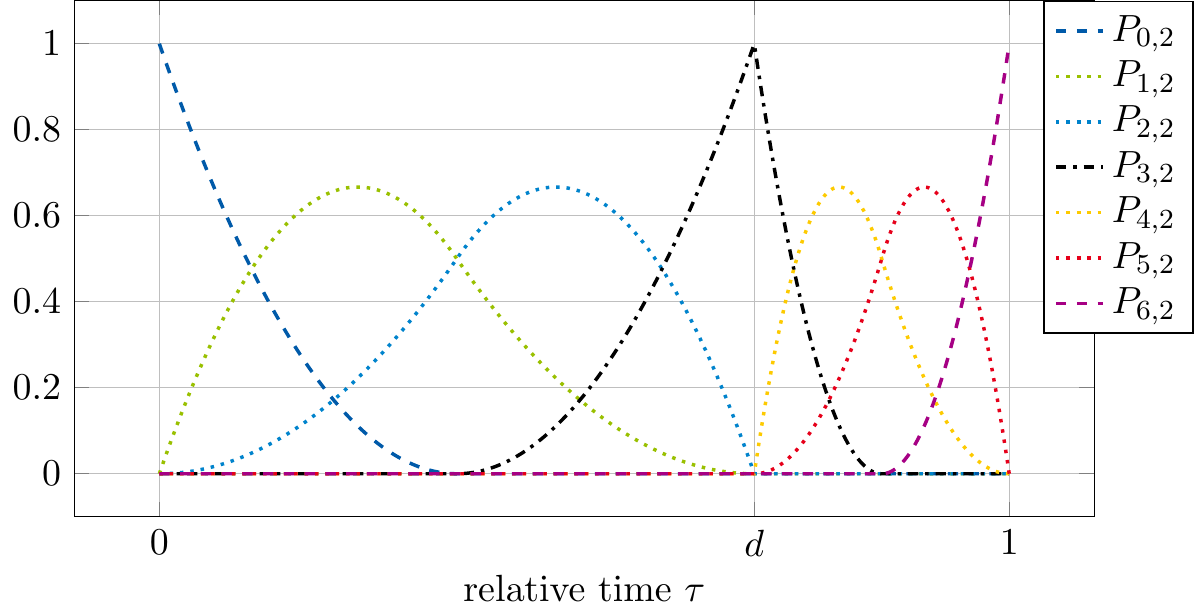}
  \caption{B-spline basis functions with degree $\degree=2$ and refinement factor $K=1$.}
  \label{fig:bSplines}
\end{figure}

\begin{theorem}[Dependency of the matrices $\cI, \cQ, \cU$ on the duty cycle]\label{the:depMat}
  Using the B-spline basis functions \eqref{equ:BSplinesSpecific}, only the matrix $\cI$ from \eqref{equ:matI} depends linearly on the duty cycle, i.e., $\cI(\du)=\frac{\cI(\du_0)-\cI(\du_1)}{\du_0-\du_1}(\du-\du_0)+\cI(\du_0)$, with $\du_0,\du_1 \in (0,1)$ and $\du_0 \neq \du_1$, and the matrices $\cQ, \cU$ from \eqref{equ:matQ}, \eqref{equ:matU}, respectively, are independent of the duty cycle.
\end{theorem}
\begin{proof}
  The set of basis functions \eqref{equ:BSplinesSpecific} can be split into three parts: The basis functions which are supported on the knots in the interval $[0,\du]$, i.e., the basis functions $\{\bSpl_{0,\degree}, \dots, \bSpl_{\degree+K-1,\degree}\}$; the basis functions on the interval $[d,1]$, i.e., $\{\bSpl_{\degree+K+1,\degree}, \dots, \bSpl_{2\degree+2K,\degree}\}$; and the remaining $C^0$ basis function which depends on knots from the entire interval $[0,1]$ , i.e. $\{\bSpl_{\degree+K,\degree}\}$.
  
  \begin{enumerate}
    \item Calculating the basis functions $\{\bSpl_{0,\degree}, \dots, \bSpl_{\degree+K-1,\degree}\}$ using the Cox-DeBoor formula leads to a recursion of the form
    \begin{equation}
      \bSpl_{i,\degree}(\xi,\du)=\frac{\xi-\gamma_{1,i} \, \du}{\gamma_{2,i} \, \du} \bSpl_{i,\degree-1}(\xi,\du) + \frac{\gamma_{3,i}\, \du-\xi}{\gamma_{4,i} \, \du} \bSpl_{i+1,\degree-1}(\xi,\du),
    \end{equation} 
    where $\gamma_{1,i}, \dots, \gamma_{4,i}$ are constants depending on $\alpha_1,\dots,\alpha_K$, but independent of the duty cycle $\du$. 
    Therefore the resulting polynomials $\bSpl_{i,\degree}(\xi)$ can be written as 
    \begin{equation}
      \bSpl_{i,\degree}(\xi,\du) = \bar{\gamma}_{\degree,i} \frac{\xi^{\degree}}{\du^{\degree}} + \bar{\gamma}_{\degree-1,i} \frac{\xi^{\degree-1}}{\du^{\degree-1}} + \dots + \bar{\gamma}_{0,i} \frac{\xi^{0}}{\du^{0}}
      \label{equ:polFirstPart}
    \end{equation}
    
    \item To calculate the B-splines $\{\bSpl_{\degree+K+1,\degree}, \dots, \bSpl_{2\degree+2K,\degree}\}$, we first redefine the knot vector for ease of notation by shifting it
    \begin{equation}
      \widehat{\Xi}_{\degree,K}=\Xi_{\degree,K}-\du=\{-\du,\dots,\du,\alpha_1 \du-\du, \dots, \alpha_K \du-\du,0,\dots,0,\beta_1 (1-\du), \dots, \beta_K (1-\du),\du,\dots,\du\}.
    \end{equation}
    The resulting B-splines are the same as with the original knot vector $\Xi_{\degree,K}$ but shifted by $-\du$. Using the Cox-DeBoor formula leads to 
    \begin{equation}
      \widehat{\bSpl}_{i,\degree}(\xi,\du)=\frac{\xi-\chi_{1,i} \, (1-\du)}{\chi_{2,i} \, (1-\du)} \widehat{\bSpl}_{i,\degree-1}(\xi,\du) + \frac{\chi_{3,i}\, (1-\du)-\xi}{\chi_{4,i} \, (1-\du)} \widehat{\bSpl}_{i+1,\degree-1}(\xi,\du),
    \end{equation} 
    The final polynomials $\widehat{\bSpl}_{i,\degree}(\xi,\du)$ can be written as
    \begin{equation}
      \widehat{\bSpl}_{i,\degree}(\xi,\du) = \bar{\chi}_{\degree,i} \frac{\xi^{\degree}}{(1-\du)^{\degree}} + \bar{\chi}_{\degree-1,i} \frac{\xi^{\degree-1}}{(1-\du)^{\degree-1}} + \dots + \bar{\chi}_{0,i} \frac{\xi^{0}}{(1-\du)^{0}}
      \label{equ:polSecondPart}
    \end{equation}
    
    \item The single basis function $\bSpl_{\degree+K,\degree}$ consists of two terms due to the knot repetition. In the Cox-DeBoor formula, the first term in the sum stems from basis function in the interval $[0,d]$, the second term stems from basis functions in the interval $[1,d]$. Therefore as well as for the other basis functions, the left term can be written as a polynomial of the form \eqref{equ:polFirstPart}, the right term can be written as a polynomial of the form \eqref{equ:polSecondPart}.
  \end{enumerate}
  Using the above knowlegde, the matrix \eqref{equ:matI} is of the form
  \begin{equation}
    \cI_{i,j}(d) = 
    \underbrace{\Ts \int\limits_0^d  \bSpl_{i,\degree}(\xi,\du) \bSpl_{j,\degree}(\xi,\du)\, \dd \xi}_{\cI_{i,j}^0 \, \du} 
    + \underbrace{\Ts \int\limits_0^{1-d} \widehat{\bSpl}_{i,\degree}(\xi,\du) \widehat{\bSpl}_{j,\degree}(\xi,\du) \dd \xi}_{\cI_{i,j}^1 (1-\du)} 
  \end{equation}
  where $\cI_{i,j}^0, \cI_{i,j}^1$ are constants. Therefore the matrix $\cI$ depends only linearly on the duty cycle.
  
  The matrices \eqref{equ:matQ} and \eqref{equ:matU} are given by
  \begin{equation}
    \cQ_{i,j} =  \underbrace{\int\limits_0^d \frac{\partial \bSpl_{i,\degree}(\xi,\du)}{\partial \xi} \bSpl_{j,\degree}(\xi,\du)\, \dd \xi}_{\cQ_{i,j}^0} 
    + \underbrace{\int\limits_0^{1-d} \frac{\partial \widehat{\bSpl}_{i,\degree}(\xi,\du)}{\partial \xi} \widehat{\bSpl}_{j,\degree}(\xi,\du) \dd \xi}_{\cQ_{i,j}^1} 
  \end{equation}
  and 
  \begin{equation}
    \cU_{i,j} =  \underbrace{\int\limits_0^d \bSpl_{i,\degree}(\xi,\du) \frac{\partial \bSpl_{j,\degree}(\xi,\du)}{\partial \du}\, \dd \xi}_{\cU_{i,j}^0} 
    +  \underbrace{\int\limits_0^{1-d} \widehat{\bSpl}_{i,\degree}(\xi,\du) \frac{\partial \widehat{\bSpl}_{j,\degree}(\xi,\du)}{\partial \du} \dd \xi}_{\cU_{i,j}^1} 
  \end{equation}
  where $\cQ_{i,j}^0$, $\cQ_{i,j}^1$, $\cU_{i,j}^0$ and $\cU_{i,j}^1$ are constants. The matrices are thus independent of the duty cycle $\du$.
\end{proof}
This is a helpful result since it allows a cheap calculation of the matrices $\cI$, $\cQ$ and $\cU$ for different duty cycles and time steps.

\section{Numerical results}\label{sec5}
To verify the proposed method and measure its efficiency, we test the method on the inverter example as depicted in Fig.~\ref{fig:buckConverterAndSol}. 
To measure the accuracy of solutions they are compared to a reference solution calculated in Simulink using PLECS. In the following we use three simulation setups: 1.) The MPDE approach; 2.) A conventional time discretization with switch detection implemented in MATLAB; 3.) A simulation in Simulink using PLECS.
The conventional approaches 2.) and 3.) exploit a-priori knowledge on discontinuities, e.g. by a pre-defined event function.

\begin{figure}
  \begin{subfigure}[c]{1\textwidth}
    \centering
    \includegraphics{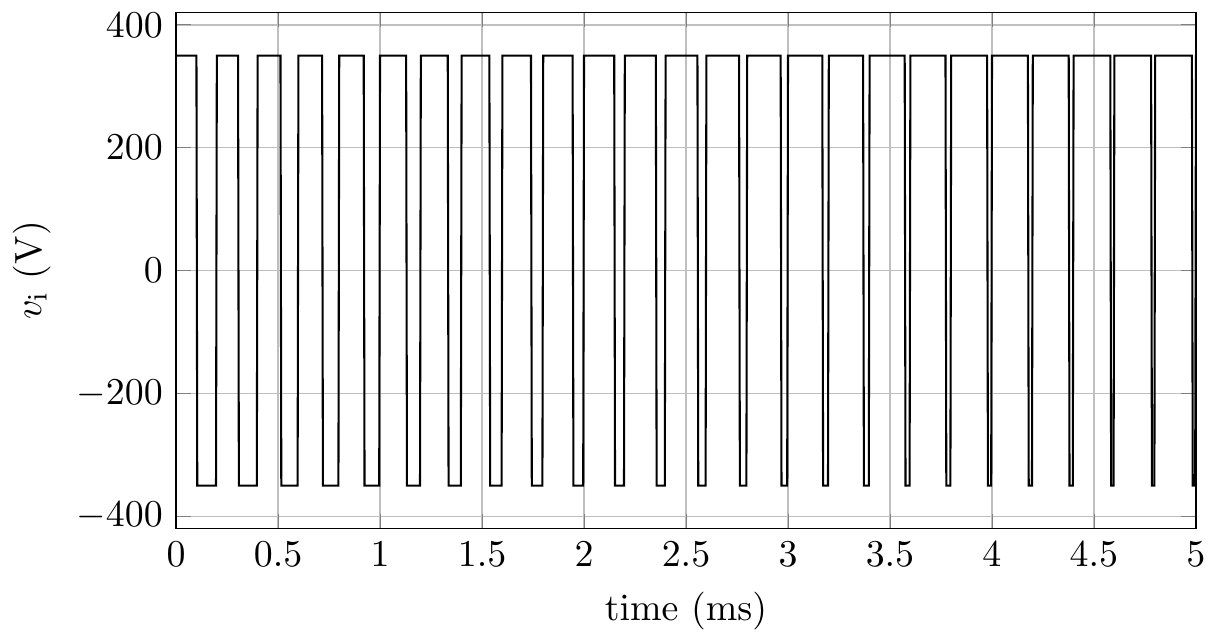}
  \end{subfigure}
  
  \vspace{2em}
  \begin{subfigure}[c]{1\textwidth}
    \centering
    \includegraphics{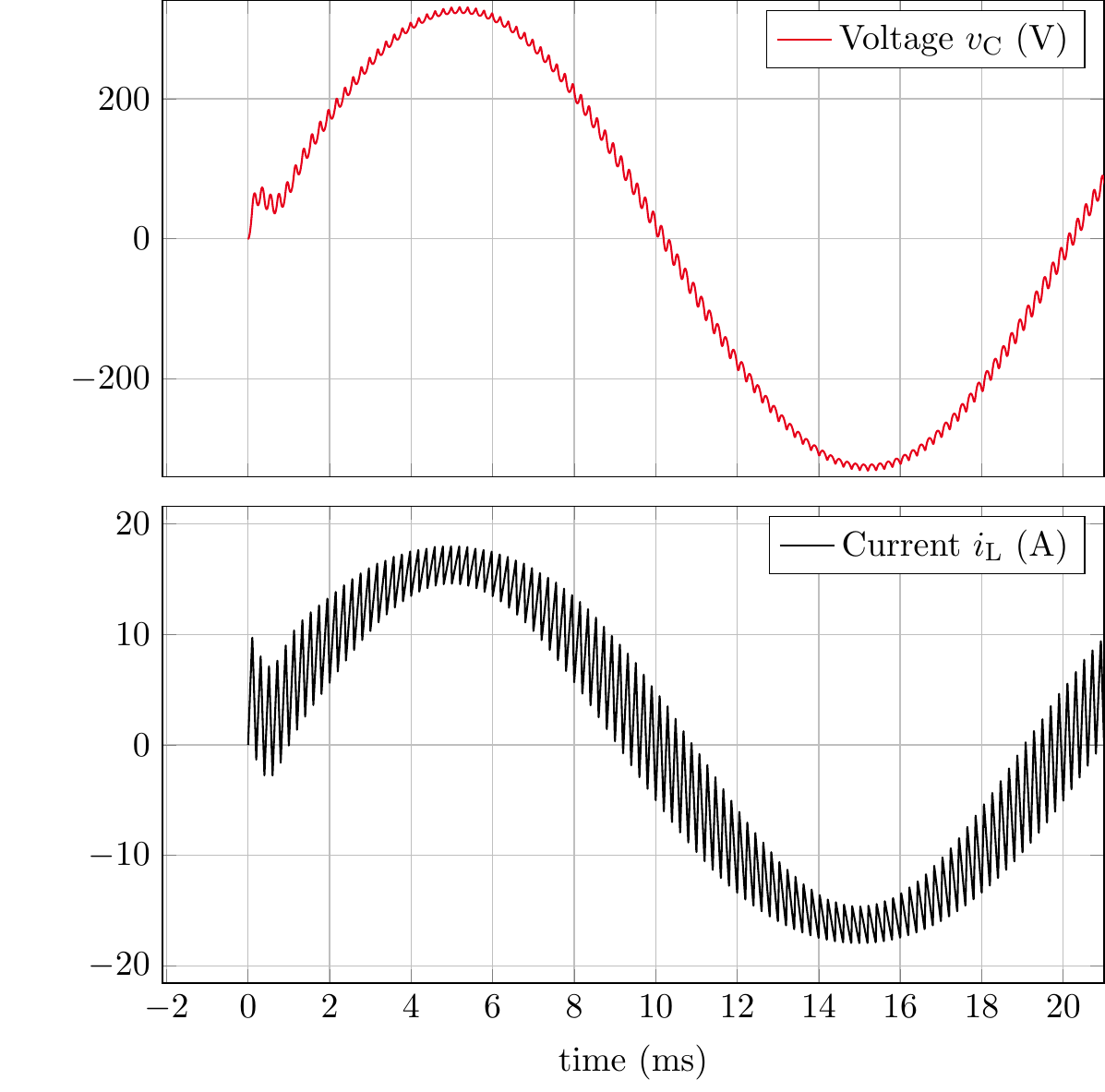}
  \end{subfigure}
  \caption{Excerpt of the excitation (top) and the solution of the inverter (bottom) using sinusoidally varying duty cycle.}
  \label{fig:excerptSolutionAndExcitation}
\end{figure}

The buck converter is described by the ordinary differential equation
\begin{equation}
  {\left[\begin{array}{cc} L & 0  \\ 0 & C  \end{array}\right]}_{}
  \frac{\dd}{\dd t}{\left[\begin{array}{cc} \iL \\ \vC \end{array}\right]}_{}
  + {\left[\begin{array}{cc} R_\L & 1  \\ -1 & 1/R  \end{array}\right]}_{}  
  \!
  {\left[\begin{array}{cc} \iL \\ \vC \end{array}\right]}_{}
  \,=\,
  { 
    \left[\begin{array}{cc} v_{\mathrm i}(t, \du(t)) \\ 0  \end{array}\right]}_{} \, ,
  \label{equ:buckConvODEs}
\end{equation}
where $L=4\,$mH (inductance of the coil), $C=10$\,$\mu$F (capacitance of the capacitor), $R_\L=10$\,m$\Omega$ (coil resistance), and $R=20\,\Omega$ (load resistance) are fixed parameters, and $\vC$, $\iL$ and $\vi$ are the voltage at the capacitor, the current through the coil and the PWM excitation, respectively. The excitation for the inverter is generated using natural sampling PWM with a sawtooth carrier. It can be written as
\begin{equation}
  \vi(t, \du(t))=\vibar \mathrm{sgn}(\du(t)-s(t)),
\end{equation}
where $s(t)=\frac{t}{\Ts}\text{ mod } 1$ is the sawtooth carrier, $\vibar$ is the peak excitation voltage, and $\mathrm{sgn}(t)$ is the sign function. The switching frequency is fixed at $\fs=1/\Ts=5\,$kHz. An excerpt of the excitation is shown in Fig.~\ref{fig:excerptSolutionAndExcitation}.

For the MPDE formulation we choose the multivariate excitation as $\vihat(t_1, t_2)=\vi(t_2, \du(t_1))$, i.e., we force the duty cycle, which is slowly varying compared to the switching of the excitation, to evolve along the time scale $t_1$ and the switching to occur along the time scale $t_2$. As it turns out the right-hand side of the ODEs after the MPDE approach, i.e., \eqref{equ:integralC} is also linearly dependent on the duty cycle and thus allows a cheap evaluation. The proof is analog to the one detailled in Section~\ref{sec4}.
The duty cycle is assumed to be sinusoidal and given by
\begin{equation}
  \du(t_1)=0.5\left(\frac{\vCbarDesired}{\vibar}\sin(2\pi\fac t_1)+1\right),
\end{equation}
where $\vCbarDesired$ is the desired peak output voltage of the converter and $\fac$ is the desired frequency of the AC output voltage. The constants are fixed to $\vibar=350\,$V, $\vCbarDesired=325\,$V (corresponding to $230\,$V effective voltage) and $\fac=50\,$Hz. The resulting inverter output, i.e., voltage at the capacitor and current through the coil, are depicted in Fig.~\ref{fig:excerptSolutionAndExcitation}. The multivariate solution of the MPDEs is depicted in Fig.~\ref{fig:buckConverterDCACoutputMVR}. The solution of the original equations \eqref{equ:origODEs} is marked as black line.

\begin{figure}
  \centering
  \includegraphics{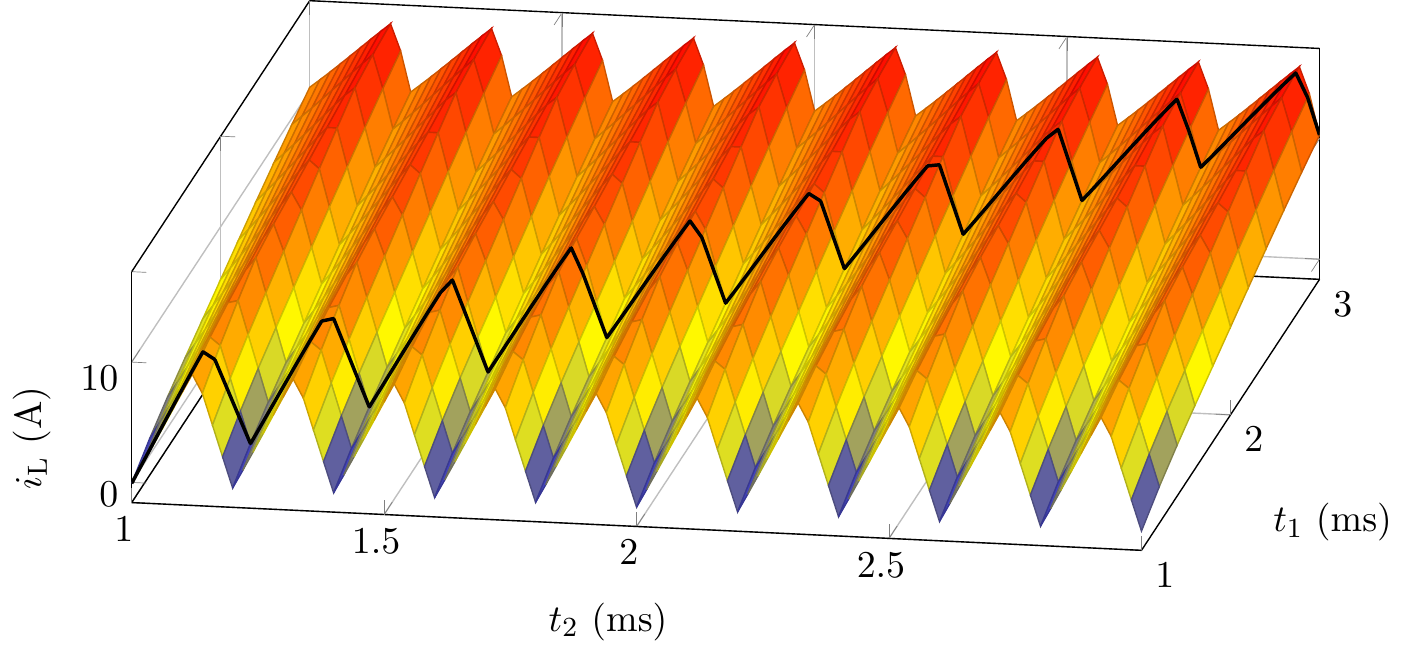}
  \caption{Excerpt of the multivariate solution of the inverter. The solution of the original ODEs is marked as black line. }
  \label{fig:buckConverterDCACoutputMVR}
\end{figure}

The MPDE approach is implemented in MATLAB and the solver \texttt{ode15s} is used for the Simulink, the time discretization and the MPDE approach simulation. For the MPDE approach, the maximum order \texttt{MaxOrder} for the solver is set to 2 whilst for the other two methods the original setting of 5 is used. To compare the solutions, the relative $\Ltwo$ error of the capacitor voltage on the simulation interval~${\Omega=[0,4\Tac]}$
\begin{equation}
  \epsilon_{\mathrm{v}} = \frac{||v_\mathrm{C,ref}(t)-v_\mathrm{C}^h(t)||_{\Ltwo(\Omega)}}{||v_\mathrm{C,ref}(t)||_{\Ltwo(\Omega)}}
  \label{equ:l2erroreps}
\end{equation}
is approximated using the mid-point quadrature rule. The error of the current through the coil $\epsilon_{\mathrm{i}}$ is defined analogously.
The reference solution $v_\mathrm{C,ref}(t)$ is calculated using a very fine tolerance of $\mathtt{abstol}=\mathtt{reltol}=10^{-12}$ and a maximum step size of $\Ts/1000$ in Simulink. 

MPDE approach simulation results are obtained for three different settings:
\begin{enumerate}
 \item Lowest order approximation: $\degree=1$ and $K=1$ (corresponds to Finite Element hat functions),
 \item Medium order approximation: $\degree=2$ and $K=1$,
 \item High order approximation: $\degree=3$ and $K=3$.
\end{enumerate}

The error $\epsilon$ of the MPDE approach using these settings is calculated for different tolerances of the time discretization and is depicted in Fig.~\ref{fig:errorMPDEapproach} for both the current through the coil and the voltage at the capacitor. The error is evaluated using 100 points per cycle. Reference solution values which are not available at the corresponding points are interpolated linearly. Fig.~\ref{fig:errorMPDEapproach} shows that the error decreases with smaller tolerances for the solver until it stagnates, which is due to the fact that the approximation of the Galerkin approach bounds the accuracy. It is also visible that with better discretization settings for the Galerkin approach, the stagnation is shifted to smaller tolerances. For the three settings introduced above we now fix the tolerances for the solver as such, that we achieve highest accuracy without wasting computational effort in the stagnation region. The corresponding tolerances are marked as dots in Fig.~\ref{fig:errorMPDEapproach} and summarized in Table~\ref{tab:tableAccuracy}. To be able to compare statistical data like the number of time steps, number of LU decompositions and number of function evaluations used by the solvers, the accuracy of conventional time discretization in MATLAB and PLECS are controlled by the tolerance $\mathtt{abstol}=\mathtt{reltol}$. The error of both with respect to the tolerance is shown in Fig.~\ref{fig:errorSimulinkAndTDSimulation} for both voltage and current. The errors corresponding approximately to the errors of the MPDE approach are listed in Table~\ref{tab:tableAccuracy} with the associated solver tolerance. Statistical data of all approaches, i.e., number of time steps, number of failed steps, number of LU decompositions, number of function evaluations and number of forward/backward substitution (solution of linear systems) are compared in Table~\ref{tab:tableSpeedup} for the settings in Table~\ref{tab:tableAccuracy}. 
For the PLECS simulation only the number of time steps is supplied. The other figures of merit can to our knowledge not be extracted from the Simulink solvers. For high solver tolerance, i.e., $\mathtt{abstol}=\mathtt{reltol}=10^{-2}$, the conventional time discretization in MATLAB is slightly more accurate than the PLECS results (see Fig.~\ref{fig:errorSimulinkAndTDSimulation}) since the \texttt{MaxStep} option is used to guarantee that no switching events are missed. This explains the considerably higher number of time steps compared to PLECS in Table~\ref{tab:tableSpeedup}, setting 1 (lowest order). For the other settings, the number of time steps between MATLAB and PLECS is almost similar.

\begin{figure}
  \begin{subfigure}[c]{1\textwidth}
    \centering
    \includegraphics{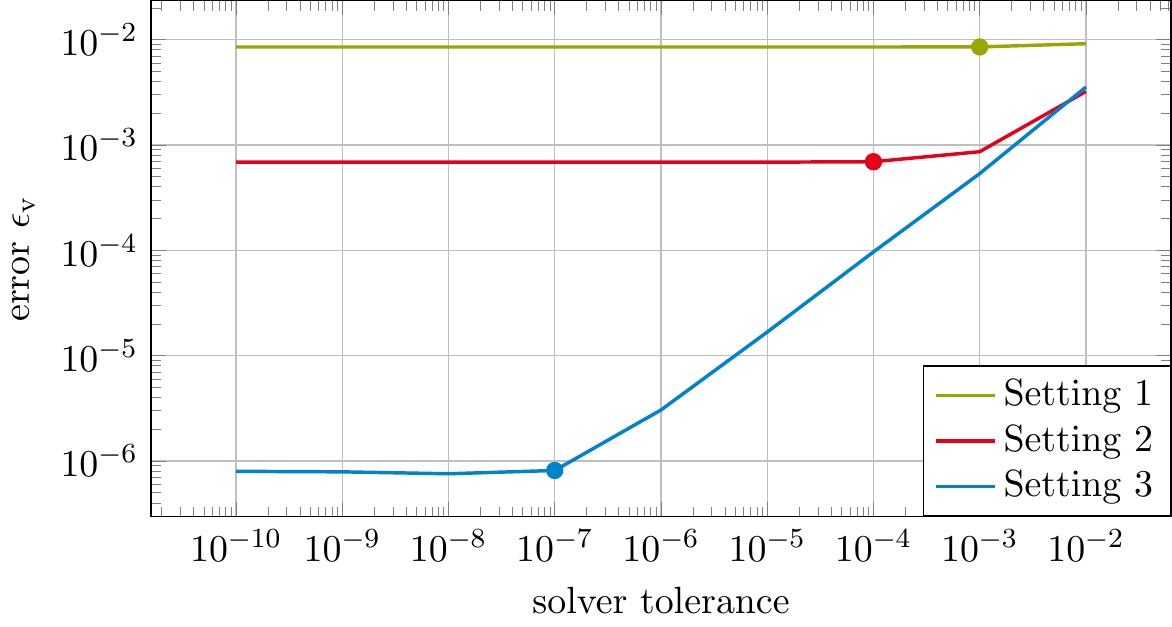}
  \end{subfigure}
  
  \vspace{1em}
  \begin{subfigure}[c]{1\textwidth}
    \centering
    \includegraphics{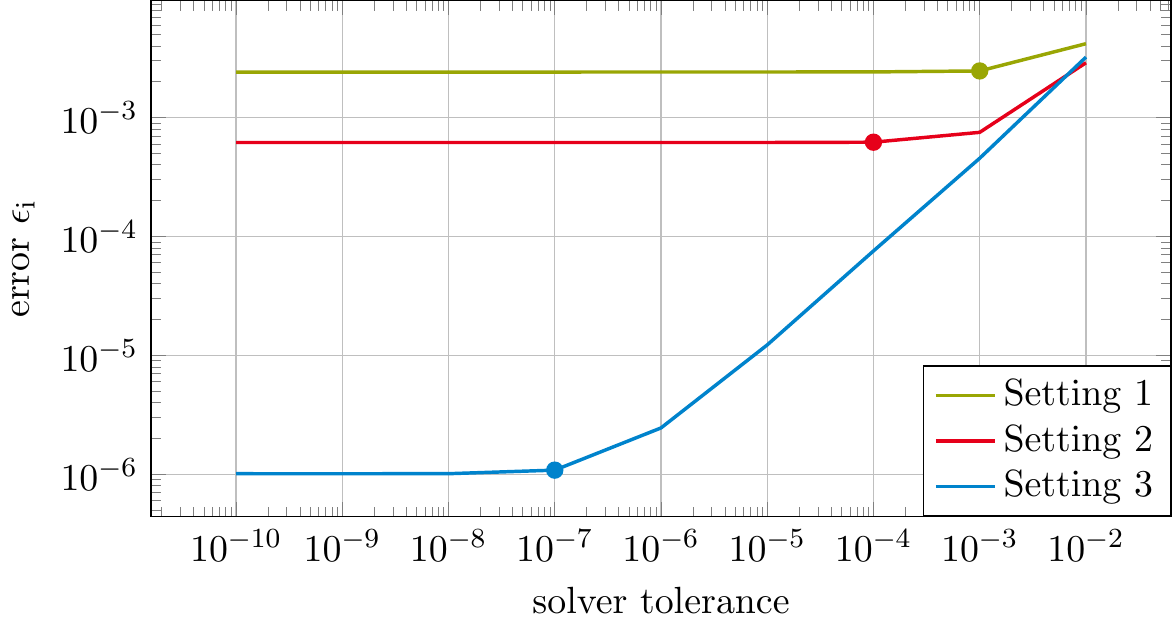}
  \end{subfigure}
  \caption{Error of the MPDE approach versus the solver tolerance for time discretization. The dots mark the tolerances chosen for the comparison to the conventional time discretization methods. (top) Error of the voltage at the capacitor. (bottom) Error of the current through the coil.}
  \label{fig:errorMPDEapproach}
\end{figure}

\begin{figure}
  \centering
  \includegraphics{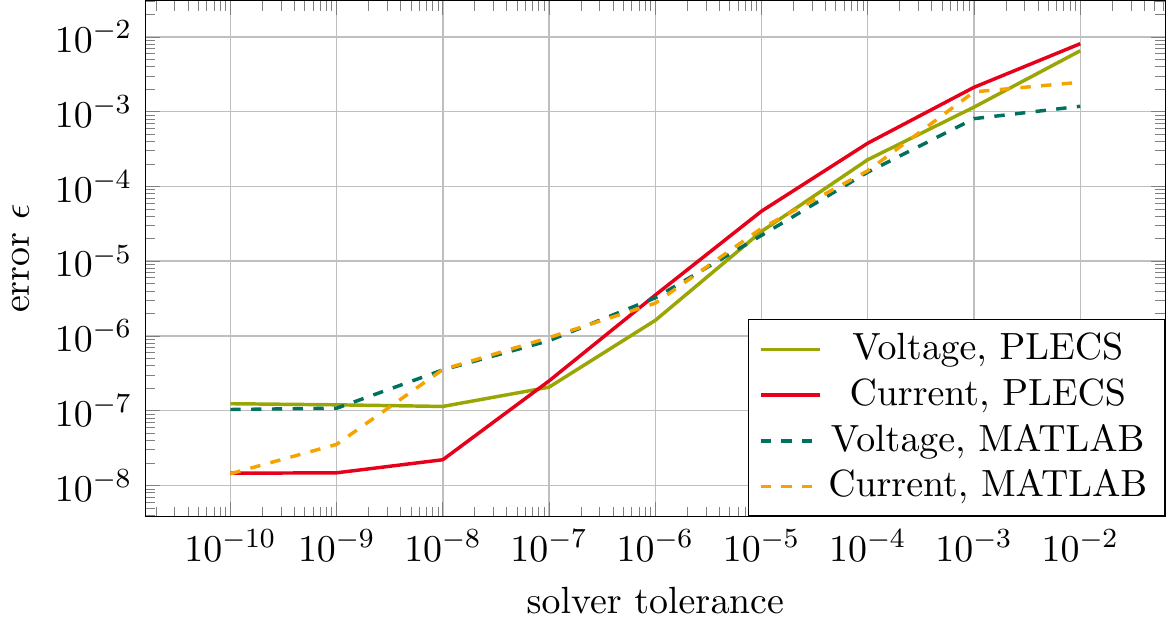}
  \caption{Error of the time discretization in MATLAB and PLECS versus the solver tolerance. }
  \label{fig:errorSimulinkAndTDSimulation}
\end{figure}

As can be seen in Table~\ref{tab:tableSpeedup} the MPDE approach is several times faster than the conventional methods with respect to the considered quantities. Especially when low or medium accuracy (setting 1 and 2) is necessary, the speedup is very high. Furthermore the efficiency of the method increases with higher switching frequency, see \cite{Pels_2017ad}. If a switching frequency is assumed to be twice as high, for instance, the conventional time discretization will likely take twice as much time since twice as many ripples have to be resolved. The MPDE approach on the other hand will need approximately the same number of time steps since the envelope does not change. As a result the speedup compared with the conventional methods would be twice as high.

It should be noted, that the actual efficiency in terms of computing time depends on the efficiency of the solver used for the solution of the equation systems. Since the equation systems in the MPDE approach are larger than the ones of the original problem, the solver will take more time. This has to be taken into account when choosing the number of basis functions for the solution expansion. The MPDE approach will be especially efficient, if the computational effort for function evaluation is higher than the effort for solving the equation systems.

\begin{center}
  \begin{table*}[t]
    \footnotesize
    \caption{Accuracy of the MPDE approach and the conventional simulation using time discretization in MATLAB and PLECS.\label{tab:tableAccuracy}}
    \centering
    \begin{tabular*}{500pt}{@{\extracolsep\fill}p{1.5cm}p{0.3cm}lp{0.3cm}lp{0.3cm}lc}
      \toprule
      &\multicolumn{2}{@{}c@{}}{\textbf{MPDE}} & \multicolumn{2}{@{}c@{}}{\textbf{MATLAB}} & \multicolumn{2}{@{}c@{}}{\textbf{PLECS}} \\\cmidrule{2-3}\cmidrule{4-5}\cmidrule{6-7}
      \textbf{Simulation setting} & \textbf{Tol.}  & \textbf{Error} voltage/current & \textbf{Tol.}  & \textbf{Error} voltage/current  & \textbf{Tol.}  & \textbf{Error} voltage/current \\
      \midrule
      lowest order & $10^{-3}$  & $8.5\cdot 10^{-3}$ / $2.5\cdot 10^{-3}$ & $10^{-2}$  & $1.2\cdot 10^{-3}$ / $2.5\cdot 10^{-3}$  & $10^{-2}$ & $6.5\cdot 10^{-3}$ / $8.1\cdot 10^{-3}$ \\
      medium order & $10^{-4}$  & $6.9\cdot10^{-4}$ / $6.2\cdot10^{-4}$  &  $10^{-4}$  &    $1.6\cdot10^{-4}$ / $1.6\cdot10^{-4}$ & $10^{-4}$ & $2.3\cdot 10^{-4}$ / $3.8\cdot 10^{-4}$ \\
      high order & $10^{-7}$  & $8.1\cdot10^{-7}$ / $1.1\cdot10^{-6}$  &  $10^{-6}$  &    $3.2\cdot10^{-6}$ / $2.7\cdot10^{-6}$   & $10^{-6}$ & $1.6\cdot 10^{-6}$ / $3.5\cdot 10^{-6}$ \\
      \bottomrule
    \end{tabular*}
  \end{table*}
\end{center}

\begin{center}
  \begin{table*}[t]
    \footnotesize
    \caption{Speedup of the MPDE approach compared to the conventional simulation using MATLAB and PLECS in different simulation settings. \label{tab:tableSpeedup}}
    \centering
    \begin{tabular*}{500pt}{@{\extracolsep\fill}p{1.8cm}p{0.7cm}p{1.4cm}p{1.2cm}p{0.7cm}p{1.4cm}p{1.2cm}p{0.7cm}p{1.4cm}p{1.2cm}@{\extracolsep\fill}}
      \toprule
      & \multicolumn{3}{@{}c@{}}{\textbf{Setting 1 (lowest order)}} & \multicolumn{3}{@{}c@{}}{\textbf{Setting 2 (medium order)}} & \multicolumn{3}{@{}c@{}}{\textbf{Setting 3 (high order)}} \\\cmidrule{2-4}\cmidrule{5-7}\cmidrule{8-10}
      \textbf{} & \textbf{MPDE}  & \textbf{MATLAB/ PLECS} & \textbf{Speedup (approx.)} & \textbf{MPDE}  & \textbf{MATLAB/ PLECS} & \textbf{Speedup (approx.)} & \textbf{MPDE}  & \textbf{MATLAB/ PLECS} & \textbf{Speedup (approx.)}  \\
      \midrule
      time steps                                          & 235   & 10533/4243   & 44/18   & 474   & 12410/9940   & 26/21     & 4516   & 18476/17508     & 4/4    \\
      failed steps                                        & 50    & 36           & 0.7     & 76    & 805          & 11        & 111    & 2641            & 24     \\
      LU decom.                                & 95    & 2535         & 27      & 157   & 4369         & 28        & 728    & 7435            & 10      \\
      solution lin. systems   & 557   & 21938        & 39      & 1062  & 27230        & 26        & 7553   & 43034           & 6      \\
      \bottomrule
    \end{tabular*}
  \end{table*}
\end{center}

\section{Conclusions} \label{sec6}
A multirate method for the efficient simulation of DC-AC switch-mode inverters has been presented. The system of equations describing the converter circuit is first formulated as a system of Multirate Partial Differential Equations (MPDEs), which allow to split the solution into components of different explicitly stated time scales. The MPDEs are efficiently solved using a combination of a Galerkin approach with B-spline basis functions for the solution expansion, and a conventional time discretization. The functionality of the proposed approach is verified on a simple inverter circuit with sinusoidal AC output. The computational efficiency is analyzed among others in terms of the number of time steps, number of LU decompositions, number of functions evaluations, compared to a conventional time discretization of the problem. It shows the high potential efficiency of the method. The method is particularly efficient in applications in which the function evaluations, i.e., the evaluations of the matrices/functions in the differential equation describing the application are computationally expensive.

\section*{Acknowledgments}
This work is supported by the ``Excellence Initiative'' of German Federal and State Governments and the Graduate School CE at TU Darmstadt.

\bibliographystyle{ieeetr} 
\bibliography{abbrv,english,library}

\end{document}